\pdfoutput=1
\documentclass[a4paper,USenglish,cleveref,autoref]{lipics-v2021}
\usepackage[utf8]{inputenc}
\usepackage{amsmath, amsthm}
\usepackage{graphicx}
\usepackage{rotating}
\usepackage[ruled,vlined,linesnumbered]{algorithm2e}
\usepackage{multirow}
\usepackage{bm}

\nolinenumbers

\hideLIPIcs
\title{An Improved Deterministic Algorithm for the Online Min-Sum Set Cover Problem} 

\author{Mateusz Basiak}{University of Wrocław}{}{}{}
\author{Marcin Bienkowski}{University of Wrocław}{}{}{}
\author{Agnieszka Tatarczuk}{University of Wrocław}{}{}{}
\authorrunning{M. Basiak, M. Bienkowski and A. Tatarczuk} 

\funding{Supported by Polish National Science Centre 
\bigskip
grant 2022/45/B/ST6/00559.\\
\emph{For the purpose of open access, the authors have applied CC-BY public
copyright licence to any Author Accepted Manuscript (AAM) version arising from
this submission.}}

\newtheorem{fact}[theorem]{Fact}
\crefname{fact}{Fact}{Facts}

\newcommand{\MSSC}{\textnormal{\textsc{Mssc}}\xspace}
\newcommand{\OPT}{\textnormal{\textsc{Opt}}\xspace}
\newcommand{\OFF}{\textnormal{\textsc{Off}}\xspace}
\newcommand{\ALG}{\textnormal{\textsc{Alg}}\xspace}
\newcommand{\LM}{\textnormal{\textsc{Lma}}\xspace}
\newcommand{\DLM}{\textnormal{\textsc{Dlm}}\xspace}
\newcommand{\MTF}{\textnormal{\textsc{Mtf}}\xspace}
\newcommand{\MTFB}{\textnormal{\textsc{Mtfb}}\xspace}
\newcommand{\STATIC}{\ensuremath{\textsc{Fixed}}\xspace}
\newcommand{\I}{\ensuremath{\mathcal{I}}}
\newcommand{\J}{\ensuremath{\mathcal{J}}}
\newcommand{\C}{\ensuremath{\mathcal{C}}}
\newcommand{\fetch}{\normalfont{\textsc{fetch}}}
\newcommand{\E}{\mathbf{E}}
\newcommand{\fa}{\ensuremath{\alpha}}
\newcommand{\fb}{\ensuremath{\beta}}
\newcommand{\fc}{\ensuremath{\gamma}}
\newcommand{\fk}{\ensuremath{\kappa}}
\newcommand{\fv}{\ensuremath{\varepsilon}}
\newcommand{\pos}{\ensuremath{\pi}}
\newcommand{\U}{\ensuremath{\mathcal{U}}}
\newcommand{\suchthat}{\ensuremath{\;\mid\;}}

\begin{document}

\maketitle

\begin{abstract}
We study the online variant of the Min-Sum Set Cover (\MSSC) problem, a
generalization of the well-known list update problem. In the \MSSC problem, an
algorithm has to maintain the time-varying permutation of the list of $n$
elements, and serve a sequence of requests $R_1, R_2, \dots, R_t, \dots$. Each
$R_t$ is a subset of elements of cardinality at most $r$. For a requested set
$R_t$, an online algorithm has to pay the cost equal to the position of the
first element from $R_t$ on its list. Then, it may arbitrarily permute its list,
paying the number of swapped adjacent element pairs. 

We present the first \emph{constructive} deterministic algorithm for this
problem, whose competitive ratio does not depend on $n$. Our algorithm is
$O(r^2)$-competitive, which beats both the \emph{existential} upper bound of
$O(r^4)$ by Bienkowski and Mucha~[AAAI '23] and the previous constructive bound
of $O(r^{3/2} \cdot \sqrt{n})$ by Fotakis et al.~[ICALP '20]. Furthermore, we
show that our algorithm attains an asymptotically optimal competitive ratio of
$O(r)$ when compared to the best fixed permutation of elements. 
\end{abstract}

\section{Introduction}

In the online Min-Sum Set Cover (\MSSC) problem~\cite{FoKKSV20,FoLiPS20}, an
algorithm has to maintain an~ordered list of elements. During runtime, an online
algorithm is given a sequence of requests $R_1, R_2, \dots, R_t, \dots$,
each being a subset of elements. When a request $R_t$ appears, an~algorithm
first pays the cost equal to the position of the first element of $R_t$ in its
list. Next, it may arbitrarily reorder the list, paying the number of swapped
adjacent elements. 

The online \MSSC finds applications in e-commerce, for maintaining an ordered
(ranked) list of all shop items (elements) to be presented to new shop
customers~\cite{DeGoMM20}. This so-called cold-start list can be updated to
reflect the preferences of already known users (where $R_t$ corresponds to the set
of items that are interesting for user $t$). Each user should be able to find at
least one interesting item close to the list beginning, as otherwise, they must
scroll down, which could degrade their overall experience. Similar phenomena
occur also in ordering results from a web search for a given
keyword~\cite{DwKuNS01,AgBrDu18}, or ordering news and advertisements.

The problem is also theoretically appealing as the natural generalization of the
well-known list update problem~\cite{SleTar85}, where all $R_t$ are singletons
(see~\cite{Kamali16} and references therein).


\subsection{Model and notation}

For any integer $\ell$, let $[\ell] = \{1, \dots, \ell\}$. 
We use $\U$ to denote the universe of $n$ elements. By permutation $\pi$ of $\U$, 
we understand a~mapping $\U \to [n]$ (from items to their list positions).
Thus, for an element $z \in \U$, $\pi(z)$ is its 
position on the list. 

An input~$\I$ to the online \MSSC problem consists of an initial permutation
$\pi_0$ of $\U$ and a~sequence of $m$ sets $R_1, R_2, \dots, R_m$. In step $t$,
an online algorithm \ALG is presented a request~$R_t$ and is charged the
\emph{access cost} $\min_{z \in R_t} \pi_{t-1}(z)$. Then, \ALG chooses a new permutation
$\pi_t$ (possibly $\pi_t = \pi_{t-1}$) paying \emph{reordering cost}
$d(\pi_{t-1}, \pi_{t})$, equal to the minimum number of swaps of adjacent
elements necessary to change permutation $\pi_{t-1}$ into $\pi_t$.\footnote{The
value $d(\pi_{t-1}, \pi_{t})$ is also equal to the number of inversions between
$\pi_{t-1}$ and $\pi_t$, i.e., number of unordered pairs $(x,y)$ such that
$\pi_{t-1}(x) < \pi_{t-1}(y)$ and $\pi_{t}(x) > \pi_{t}(y)$.}

We emphasize that the choice of~$\pi_t$
made by \ALG has to be performed without the knowledge of future sets $R_{t+1},
R_{t+2}, \dots$ and also without the knowledge of the sequence length $m$. 
We use $r$ to denote the maximum cardinality of requested sets $R_t$.


\subsection{Benchmarks}

In the following, for an input $\I$ and an~algorithm~$A$, we use $A(\I)$ to
denote the total cost of $A$ on $\I$. To measure the effectiveness of online
algorithms, we use the standard notion of competitive ratio, but we generalize
it slightly, for more streamlined definitions of particular scenarios.

We say that an online algorithm \ALG is $c$-competitive \emph{against class $\C$
of offline algorithms} if there exists a constant $\xi$, such that for any input
$\I$ and any offline algorithm $\OFF \in \C$, it holds that $\ALG(\I) \leq c
\cdot \OFF(\I) + \xi$. If $\xi = 0$, then \ALG is called \emph{strictly}
competitive. The competitive ratio of $\ALG$ against class $\C$ is the infimum
of values of $c$, for which \ALG is $c$-competitive against this class. For
randomized algorithms, we replace the cost $\ALG(\I)$ with its expected value
$\E[\ALG(\I)]$. We consider three scenarios.

\subparagraph{Dynamic scenario.}

In the dynamic scenario, the considered class $\C$ contains all possible offline
algorithms, in particular those that adapt their permutation dynamically during
runtime. This setting is equivalent to the traditional competitive
ratio~\cite{BorEl-98}, where an online algorithm is compared to the optimal
offline solution \OPT. This scenario is the main focus of this paper.

\subparagraph{Static scenario.}

Previous papers focused also on a simpler \emph{static scenario}, where the
considered class of algorithms $\STATIC$ contains all possible $n!$ fixed
strategies: an algorithm from class \STATIC starts with its list ordered
according to a fixed permutation and never changes it~\cite{FoKKSV20}. (In this
scenario, the starting permutation of an~online algorithm and an~offline
solution are different.) Note that such an offline algorithm incurs no
reordering cost, and pays access costs only. It is worth mentioning that there
exist inputs $\I$, for which $\min_{A \in \STATIC} A(\I) = \Omega(n) \cdot
\OPT(\I)$~\cite{FoKKSV20}. 

\subparagraph{Learning scenario.}

The static scenario can be simplified further, by assuming that reordering
incurs no cost on \ALG. We call such a setting \emph{learning scenario}. Clearly,
the competitive ratios achievable in the learning scenario are not larger than
those for the static scenario, which are in turn not larger than those in 
the dynamic scenario.


\subsection{Previous results}

Below, we discuss known results for the \MSSC problem in the three scenarios
described above (dynamic, static, and learning). Furthermore, we make a distinction
between ratios achievable for polynomial-time algorithms and algorithms whose
runtime per request is not restricted. The lower and upper bounds described
below are also summarized in \autoref{tab:resultsTable}.

\subparagraph{Lower bounds.}

Feige et al.~\cite{FeLoTe04} studied the \emph{offline} variant of \MSSC (where
all~$R_t$'s are given upfront and an algorithm has to compute a fixed
permutation minimizing the access cost). They show that unless $\textsf{P} =
\textsf{NP}$, no offline algorithm can achieve an approximation ratio better
than $4$. This result implies a lower bound of $4$ on the competitive ratio of
any polynomial-time online algorithm (assuming $\textsf{P} \neq \textsf{NP}$) as such
a solution can be used to solve the offline variant as well. We note that
$4$-approximation algorithms for the offline variant are known as
well~\cite{FeLoTe04,BaBHST98}.

The online version of \MSSC was first studied by Fotakis et al.~\cite{FoKKSV20}. They
show that no deterministic algorithm can achieve a ratio better than 
$\Omega(r)$ even in the learning scenario. This yields the same lower bound for 
the remaining scenarios as well.


\newcommand{\minitab}[2][l]{\begin{tabular}{#1}#2\end{tabular}}
\begin{table}[t]
        \centering
        \begin{tabular}{|c|c|c|c|c|c|}
            \cline{3-6}
            \multicolumn{2}{c|}{\multirow[c]{2}[0]{*}[0pt]{}}& \multicolumn{2}{c|}{randomized} & \multicolumn{2}{c|}{deterministic}\\
            \cline{3-6}
            \multicolumn{2}{c|}{}&LB&UB&LB&UB\\
            \cline{3-6}
            \hline
            \multirow[c]{2}[0]{*}[0pt]{learning}&unrestr. & $1$ & $1+\fv$  & $\Omega(r)$~\cite{FoKKSV20} & $O(r)$ \\
            \cline{2-6}
            &poly-time& $4$ \cite{FeLoTe04} & $11.713$ \cite{FoLiPS20} & $\Omega(r)$ & $O(r)$ \cite{FoLiPS20} \\
            \hline 
            \hline
            \multirow[c]{3}[0]{*}[0pt]{static}&unrestr. & $1$ & $1+\fv$ \cite{BluBur00} & $\Omega(r)$ & $O(r)$ \cite{FoKKSV20} \\
            \cline{2-6}
            &\multirow[c]{2}[0]{*}[0pt]{poly-time}
            &\multirow[c]{2}[0]{*}[0pt]{$4$} 
            &$O(r^2)$
            &\multirow[c]{2}[0]{*}[0pt]{$\Omega(r)$}
            & $\exp(O(\sqrt{\log n \cdot \log r}))$ \cite{FoKKSV20} \\
            & & & $\bm{O(r)}$ \textbf{(\autoref{thm:DLMSO})} & & $\bm{O(r)}$ \textbf{(\autoref{thm:DLMSO})} \\
            \hline 
            \hline
            \multirow[c]{4}[0]{*}[0pt]{dynamic}
            &\multirow[c]{2}[0]{*}[0pt]{unrestr.}
            &\multirow[c]{2}[0]{*}[0pt]{$1$}
            &\multirow[c]{2}[0]{*}[0pt]{$O(r^2)$}
            &\multirow[c]{2}[0]{*}[0pt]{$\Omega(r)$}
            &$O(r^4)$~\cite{BieMuc23} \\
            & & & & & $\bm{O(r^2)}$ \textbf{(\autoref{thm:DLMDO})} \\
            \cline{2-6}
            &\multirow[c]{2}[0]{*}[0pt]{poly-time}
            &\multirow[c]{2}[0]{*}[0pt]{$4$} 
            &\multirow[c]{2}[0]{*}[0pt]{$O(r^2)$ \cite{BieMuc23}}
            &\multirow[c]{2}[0]{*}[0pt]{$\Omega(r)$}
            &$O(r^{3/2} \cdot \sqrt{n})$ \cite{FoKKSV20}  \\
            & & & & & $\bm{O(r^2)}$ \textbf{(\autoref{thm:DLMDO})} \\
            \hline 
        \end{tabular}
        \caption{Known lower and upper bounds for the online~\MSSC problem for three scenarios 
        (dynamic, static and learning), for polynomial-time and computationally unrestricted
        algorithms. 
        Unreferenced results are trivial consequences of other results.
        The ratios proved in this paper are in bold. }
    \label{tab:resultsTable}
\end{table}


\subparagraph{Asymptotically tight upper bounds.}

For the static scenario, the randomized $(1+\fv)$-competitive solution (for any
$\fv > 0$) follows by combining multiplicative weight updates~\cite{LitWar94,ArHaKa12} 
with the techniques of Blum and
Burch~\cite{BluBur00} designed for the metrical task systems. This approach has
been successfully derandomized by Fotakis et al.~\cite{FoKKSV20}, who gave
a~deterministic solution with an asymptotically optimal ratio of $O(r)$. These
algorithms clearly also work in the learning scenario. However, in both 
scenarios, they require exponential time as they keep track
of access costs for all possible $n!$ permutations. 

Fotakis et al.~\cite{FoLiPS20}~showed that in the learning scenario,
one can maintain a sparse representation of all permutations and achieve asymptotically
optimal results that work in polynomial time: a randomized $O(1)$-competitive algorithm 
and a deterministic \mbox{$O(r)$-competitive} one.

\subparagraph{Non-tight upper bounds.}

Much of the effort in the previous papers was devoted to creating algorithms for
the dynamic scenario with low competitive ratios. For $r = 1$, a simple
\textsc{Move-To-Front} policy that moves the requested element to the first
position is $O(1)$-competitive Perhaps surprisingly, however, the competitive
ratios of many of its natural generalizations were shown to be not better than
$\Omega(n)$~\cite{FoKKSV20}.

Fotakis et al.~\cite{FoKKSV20} gave an online deterministic $O(r^{3/2} \cdot
\sqrt{n})$-competitive algorithm \textsc{Move-All-Equally} (\textsc{Mae}) and
showed that its analysis is almost tight. They provided a~better bound for the
performance of \textsc{Mae} in the static scenario: in such a setting its
competitive ratio is $\exp(O(\sqrt{\log n \cdot \log r}))$~\cite{FoKKSV20}.

For randomized solutions, this result was improved by Bienkowski and Mucha
\cite{BieMuc23}, who gave an $O(r^2)$-competitive \emph{randomized} algorithm \LM
(for the dynamic scenario). Their analysis holds also against so-called
adaptive-online adversaries, and therefore, by the reduction of~\cite{BeBKTW94},
it implies the \emph{existence} of a~deterministic $O(r^4)$-competitive
algorithm. While using the techniques of Ben-David et~al.~\cite{BeBKTW94},
the construction of such an algorithm is possible, it was not done explicitly, and
furthermore, a straightforward application of these techniques would lead to
a~huge running time.


\subsection{Our contribution}

In \autoref{sec:algorithm}, we present the first constructive deterministic
algorithm whose competitive ratio in the dynamic scenario is a function of $r$
only (and does not depend on $n$). Our algorithm, dubbed
\textsc{Deterministic-And-Lazy-Move-To-Front} (\DLM), runs in polynomial time,
and we analyze its performance both in static and dynamic scenarios. 

In the static scenario, studied in \autoref{sec:static}, \DLM attains the
optimal competitive ratio of~$O(r)$, improving over the $\exp(O(\sqrt{\log n
\cdot \log r}))$-competitive solution by Fotakis et al.~\cite{FoKKSV20} and
matching $O(r)$ bound achieved by the exponential-time algorithm
of~\cite{FoKKSV20}.

In the dynamic scenario, studied in \autoref{sec:dynamic}, we show that \DLM is
$O(r^2)$-competitive. As $r \leq n$, this bound is always better than the
existing non-constructive upper bound of $O(r^4)$~\cite{BieMuc23} and the
polynomial-time upper bound of  $O(r^{3/2} \cdot \sqrt{n})$~\cite{FoKKSV20}. Our
analysis is asymptotically tight: in \autoref{sec:lower}, we show that the ratio
of $O(r^2)$ is best possible for the whole class of approaches that includes 
\DLM.

Finally, as the learning scenario is not harder than the static one, \DLM is
$O(r)$-competitive also there. While an upper bound of $O(r)$ was already known
for the learning scenario~\cite{FoLiPS20}, our algorithm uses a vastly
different, combinatorial approach and is also much faster.

Our deterministic solution is inspired by the ideas for the randomized algorithm 
\LM of~\cite{BieMuc23} and can be seen as its derandomization, albeit 
with several crucial differences. 
\begin{itemize}
\item We simplify their approach as we solve the \MSSC problem directly,
while they introduced an intermediate exponential caching problem. 
\item We update item budgets differently, which allows us to obtain an optimal ratio 
for the static scenario ($\LM$ is not better in the static scenario 
than in the dynamic one). 
\item 
Most importantly, \cite{BieMuc23} uses randomization to argue that \LM makes bad
choices with a~small probability. In this context, bad choices mean moving
elements that \OPT has near the list front towards the list tail in the 
solution of \LM. In the deterministic approach, we obviously cannot prove the same
claim, but we show that it holds on the average. Combining this claim
with the amortized analysis, by ``encoding'' it in the additional potential 
function $\Psi$, is the main technical contribution of our paper.
\end{itemize}

\section{Our algorithm DLM} 
\label{sec:algorithm}


\DLM maintains a budget $b(z)$ for any element $z \in \U$. At the beginning of
an input sequence, all budgets are set to zero.

In the algorithm description, we skip step-related subscripts when it does not
lead to ambiguity, and we simply use $\pos(z)$ to denote the \emph{current}
position of element $z$ in the permutation of \DLM. 

\SetAlgorithmName{Routine}{routine}{List of routines}
\begin{algorithm}[tb]
\caption{\fetch(z), where $z$ is any element}
\label{alg:fetch}
\For{$i = \pos(z), \dots, 3, 2$}{
    \text{swap elements on positions $i$ and $i-1$} \\
}
$b(z) \gets 0$
\end{algorithm} 

At certain times, \DLM moves an element $z$ to the list front. It does so using a
straightforward procedure \fetch(z) (cf.~Routine~\ref{alg:fetch}). It uses
$\pos(z)-1$ swaps that move $z$ to the first position, and increment the positions
of all elements that preceded $z$. Next, it resets the budget of $z$ to zero. 

Assume now that \DLM needs to serve a request $R = \{x, y_1, y_2, \dots, y_{s-1}
\}$ (where $s \leq r$ and $\pos(x) < \pos(y_i)$ for all $y_i$). Let $\ell =
\pos(x)$. \DLM first executes routine $\fetch(x)$. Afterward, it performs a lazy
counterpart of moving elements $y_i$ towards the front: it increases their
budgets by $\ell / s$. Once a budget of any element reaches or exceeds its
current position, \DLM fetches it to the list front. The pseudocode of \DLM on
request $R$ is given in \Cref{alg:dlma}.

\SetAlgorithmName{Algorithm}{algorithm}{List of algorithms}
\begin{algorithm}[tb]
\caption{A single step of \textsc{Deterministic-Lazy-Move-All-To-Front} (\DLM)\\
\textbf{Input:} request $R = \{x, y_1, y_2, \dots, y_{s-1} \}$, where $s \leq r$ and
$\pos(x) \leq \pos(y_i)$ for $i \in [s-1]$, current permutation $\pi$ of elements} 
\label{alg:dlma}
$\ell \gets \pos(x)$ \label{line:dlma_first}\\
\textbf{execute} $\fetch(x)$ \\
\For{$i = 1, 2, \dots, s-1$}{
    $b(y_i) \leftarrow b(y_i) + \ell / s$ \label{line:dlma_budget_increase}
}
\While{\textnormal{exists $z$ such that $b(z) \geq \pos(z)$}}{  \label{line:dlma_last_line_2}
    \textbf{execute} $\fetch(z)$  \label{line:dlma_last_line_3}
}
\end{algorithm} 


\section{Basic properties and analysis framework}

We start with some observations about elements' budgets; in particular, we show
that \DLM is well defined, i.e., it terminates.

\begin{lemma}
\label{lem:budgets_controlled}
\DLM terminates after every request. 
\end{lemma}

\begin{proof}
Let $C = \{ z \in \U \suchthat b(z) \geq \pos(z) \}$. It suffices to show that
the cardinality of $C$ decreases at each iteration of the while loop in
\autoref{line:dlma_last_line_2} of \Cref{alg:dlma}. To this end, observe that in
each iteration, we execute operation $\fetch(z)$ for some $z \in C$. In effect,
the budget of $z$ is set to $0$, and thus $z$ is removed from $C$. The positions
of elements that preceded $z$ are incremented without changing their
budget: they may only be removed from $C$ but not added to it.
\end{proof}

\begin{observation}
\label{obs:budget_invariant}
Once \DLM finishes list reordering in a given step,
$b(z) < \pos(z)$ for any element $z \in \U$.
Moreover, $b(z) < (3/2) \cdot \pos(z)$ also during list reordering.
\end{observation}

\begin{proof}
Once the list reordering terminates, by \autoref{lem:budgets_controlled} and the while loop
in Lines \ref{line:dlma_last_line_2}--\ref{line:dlma_last_line_3}
of~\Cref{alg:dlma}, $b(z) < \pos(z)$ for any element $z$.

Within a step, the budgets are increased only for elements $y_i \in R$, i.e., only
when $s \geq 2$. The budget
of such an element $y_i$ is increased from at most $\pos(y_i)$ by $\pos(x) / s \leq \pos(x) / 2 <
\pos(y_i) / 2$, i.e., its resulting budget is smaller than $(3/2) \cdot \pos(y_i)$.
\end{proof}


\subsection{Amortized analysis}
\label{sec:amortized}

In our analysis, we compare the cost in a single step of \DLM to the
corresponding cost of an~offline solution \OFF. For a more streamlined analysis
that will yield the result both for the static and dynamic scenarios, we split
each step into two stages. In the first stage, both \DLM and \OFF pay their
access costs, and then \DLM reorders its list according to its definition. In
the second stage, \OFF reorders its list. Note that the second stage exists
only in the dynamic scenario. 

We use $\pos$ and $\pos^*$ to denote the current permutation of \DLM and \OFF,
respectively. We introduce two potential functions $\Phi$ and $\Psi$, whose
values depend only on $\pos$ and~$\pos^*$.

In \autoref{sec:static}, we show that in the first stage of any step, it holds that 
\begin{equation}
    \label{eq:stage1_amortized}
    \Delta \DLM + \Delta \Phi + \Delta \Psi \leq O(r) \cdot \Delta \OFF.
\end{equation}
where $\Delta \DLM$, $\Delta \OFF$, $\Delta \Phi$, and $\Delta \Psi$ denote
increases of the costs of $\DLM$ and $\OFF$ and the increases of values of
$\Phi$ and $\Psi$, respectively. Relation \eqref{eq:stage1_amortized} summed
over all $m$ steps of the input sequence yields the competitive ratio of $O(r)$
of \DLM in the static scenario (where only the first stage is present).

In \autoref{sec:dynamic}, we analyze the performance of \DLM in the dynamic scenario.
We say that an~offline algorithm \OFF is \emph{MTF-based} if, for any request,
it moves one of the requested elements to the first position of the list and
does not touch the remaining elements. We define a class \MTFB of all MTF-based
offline algorithms. We show that in the second stage of any step, it holds that 
\begin{equation}
    \label{eq:stage2_amortized}
    \Delta \DLM + \Delta \Phi + \Delta \Psi \leq O(r^2) \cdot \Delta \OFF.
\end{equation}
for any $\OFF \in \MTFB$. Now, summing relations \eqref{eq:stage1_amortized} and
\eqref{eq:stage2_amortized} over all steps in the input yields that \DLM is
$O(r^2)$-competitive against the class \MTFB. We conclude by arguing that there
exists an MTF-based algorithm $\OFF^*$ which is a $4$-approximation of the
optimal solution~$\OPT$.

\subsection{Potential function}

To define potential functions, we first 
split $\pos(z)$ into two summands, $\pos(z) = 2^{p(z)}+q(z)$, such that 
$p(z)$ is a non-negative integer, and $q(z) \in \{0, \dots, 2^{p(z)}-1\}$.
We split $\pos^*(z)$ analogously as $\pos^*(z) = 2^{p^*(z)}+q^*(z)$.

We use the following parameters:
$\fa = 2$, $\fc = 5r$, $\fb = 7.5r + 5$, and $\fk = \lceil \log (6 \fb) \rceil$. 
Our analysis does not depend on the specific values of these parameters, 
but we require that they satisfy the following relations.

\begin{fact}
\label{fact:constants}
Parameters $\fa$, $\fb$ and $\fc$ satisfy the following relations: 
$\fa \geq 2$,
$\fc \geq (3 + \fa) \cdot r$,
$\fb \geq 3 + \fa + (3/2) \cdot \fc$.
Furthermore, $\fk$ is an integer satisfying $2^{\fk} \geq 6 \fb$.
\end{fact}
For any element $z$, we define its potentials
\begin{equation}
\label{eq:def_phi}
    \Phi_z = \begin{cases}
        \fa \cdot b(z) & \text{if $p(z) \leq p^*(z) + \fk$,} \\
        \fb \cdot \pos(z) - \fc \cdot b(z) 
            & \text{if $p(z) \geq p^*(z) + \fk + 1$.} 
    \end{cases}
\end{equation}
\begin{equation}
\label{eq:def_psi}
    \Psi_z = \begin{cases}
        0 & \text{if $p(z) \leq p^*(z) + \fk - 1$,} \\
        2 \fb \cdot q(z) & \text{if $p(z) \geq p^*(z) + \fk$.} 
    \end{cases}
\end{equation}
We define the total potentials as 
$\Phi = \sum_{z \in \U} \Phi_z$ and $\Psi = \sum_{z \in \U} \Psi_z$. 

\begin{lemma}
\label{lem:non_negative}
At any time and for any element $z$, $\Phi_z \geq 0$ and $\Psi_z\geq 0$.
\end{lemma}

\begin{proof}
The relation $\Psi_z \geq 0$ follows trivially from \eqref{eq:def_psi}.
By \cref{fact:constants}, $\fb \geq (3/2) \cdot \fc$.
This, together with \autoref{obs:budget_invariant}, 
implies that $\Phi_z \geq 0$. 
\end{proof}


\subsection{Incrementing elements positions}

We first argue that increments of elements' positions induce small changes in
their potentials. Such increments occur for instance when \DLM fetches an
element $z$ to the list front: all elements that preceded $z$ are shifted by one
position towards the list tail. We show this property for the elements on the list
of \DLM first and then for the list of \OFF.

We say that an element $w$ is \emph{safe} if $p(w) \leq p^*(w) + \fk - 1$ 
and \emph{unsafe} otherwise. Note that for a safe element $w$, it holds 
that $\pi(w) \leq 2^{p(w)+1} \leq 2^{\fk} \cdot 2^{p^*(w)} \leq 2^{\fk} \cdot \pos^*(w)
= O(r) \cdot \pos^*(w)$, i.e., its position on the list of \DLM 
is at most $O(r)$ times greater than on the list of \OFF.

\begin{lemma}
\label{lem:elem_back_change}
Assume that the position of an element $w$ on the list of \DLM increases by $1$. 
Then, $\Delta \Phi_w + \Delta \Psi_w \leq 0$ if $w$ was safe before the movement 
and $\Delta \Phi_w + \Delta \Psi_w \leq 3 \beta$ otherwise.
\end{lemma}

\begin{proof}
By $\pi(w)=2^{p(w)}+q(w)$ and $\pi'(w) = \pi(w)+1=2^{p'(w)}+q'(w)$ we denote the
positions of $w$ before and after the movement, respectively.

Assume first that $w$ was safe before the movement. As $p'(w) \leq p(w)+1 \leq
p^*(w) + \fk$, $\Delta \Phi_z = \alpha \cdot b(z) - \alpha \cdot b(z) = 0$.
Furthermore, either $p'(w) = p(w)$, and then $\Delta \Psi_w = 0$ trivially, or
$p'(w) = p(w) + 1$, and then $q'(w) = 0$. In the latter case $\Delta \Psi_w = 2
\beta \cdot q'(z) - 0 = 0$ as well. This shows the first part of the lemma. 

Assume now that $w$ was unsafe ($p(w) \geq p^*(w) + \fk$) before the movement. 
We consider two cases.

\begin{itemize}
\item $p(w) = p^*(w) + \fk$ and $p'(w) = p(w) + 1$.

    It means that $q(w) = 2^{p(w)}-1$ and $q'(w) = 0$. 
    Then, 
    \begin{align*}
        \Delta \Phi_w & = \beta \cdot \pi'(z) - \gamma \cdot \beta(z) - \alpha \cdot \beta(z)
            \leq \beta \cdot \pi'(z) = \beta \cdot 2^{p'(w)} = 2 \beta \cdot 2^{p(w)}
            & \text{\;and} \\
        \Delta \Psi_w & = 2 \beta \cdot q'(z) - 2 \beta \cdot q(z) = - 2 \beta \cdot (2^{p(w)}-1)
            = -2 \beta \cdot 2^{p(w)} + 2 \beta.
    \end{align*}
    That is, the large growth of $\Phi_w$ is compensated by the drop of $\Psi_w$, 
    i.e, $\Delta \Phi_w + \Psi_w \leq 2 \beta$. 

\item $p(w) > p^*(w) + \fk$ or $p'(w) = p(w)$.

In such case, there is no case change in the definition of $\Phi_w$,
i.e, 
\[
    \Delta \Phi_w = 
    \begin{cases}
        \alpha \cdot b(w) - \alpha \cdot b(w) = 0 
            & \text{if $p(w) \leq p^*(w) + \fk$,} \\
        (\beta \cdot \pos'(w) - \gamma \cdot b(w)) - (\beta \cdot \pos(w) - \gamma \cdot b(w)) = \beta
            & \text{otherwise.}
    \end{cases}
\]

Furthermore, as $q'(w) \leq q(w) + 1$, 
$\Delta \Psi(z) = 2 \beta \cdot q'(w) - 2 \beta \cdot q(w) \leq 2\beta$.
Together, $\Delta \Phi_w  + \Delta \Psi_w \leq \beta + 2 \beta = 3\beta$.
\qedhere
\end{itemize}
\end{proof}

\begin{lemma}
\label{lem:opt_increments_elems}
Assume that the position of an element $w$ on the list of \OFF increases by $1$. 
Then, $\Delta \Phi_w \leq 0$ and $\Delta \Psi_w \leq 0$.
\end{lemma}

\begin{proof}
Note that $p^*(w)$ may be either unchanged 
(in which case the values of $\Phi_w$ and $\Psi_w$ remain intact)
or it may be incremented. We analyze the latter case.

By \eqref{eq:def_phi}, the definition of $\Phi_w$, the value of  
$\Phi_w$ may change only
if $p^*(w)$ is incremented
from $p(w) + \fk - 1$ to $p(w) + \fk$. 
In such case, 
\begin{align*}
    \Delta\Psi_w 
    & = \alpha \cdot b(w) - \beta \cdot \pos(w) + \gamma \cdot b(w) \\
    & \leq (\alpha + \gamma - \beta) \cdot \pi(w) 
        && \text{(by~\autoref{obs:budget_invariant})} \\
    & \leq 0.
        && \text{(by~\cref{fact:constants})}
\end{align*}

By \eqref{eq:def_psi}, the definition of $\Psi_w$, the value of  
$\Psi_w$ may change only if $p^*(w)$ is incremented
from $p(w) + \fk$ to $p(w) + \fk + 1$. 
In such case, $\Delta\Psi_w = - 2 \beta \cdot q(w) \leq 0$.
\end{proof}


\section{Analysis in the static scenario}
\label{sec:static}

As described in \autoref{sec:amortized}, in this part, we focus on the 
amortized cost of \DLM in the first stage of a step, i.e., where 
\DLM and \OFF both pay their access costs and then \DLM reorders its list.

\begin{lemma}
\label{lem:fetch_cost}
Whenever $\DLM$ executes operation $\fetch(z)$, it holds that $\Delta \DLM +
\Delta \Psi + \sum_{w \neq z} \Delta \Phi_w \leq 2 \cdot \pos(z)$.
\end{lemma} 

\begin{proof}
As defined in Routine~\ref{alg:fetch}, the cost of operation $\fetch(z)$ is
$\Delta \DLM = \pos(z) - 1 < \pos(z)$.
We first analyze the potential changes of elements from
set $K$ of $\pi(z)-1$ elements that originally preceded $z$.

Let $K' = \{ w \in K \suchthat \pi^*(w) \leq 2^{p(z) - \fk+1} \}$.
Observe that any $w \in K \setminus K'$ satisfies
$\pi^*(w) > 2^{p(z) - \fk + 1}$, which implies 
$p^*(w) \geq p(z) - \fk + 1 \geq p(w) - \fk + 1$, and thus $w$ is safe. 
Thus, among elements of $K$, only elements from $K'$ can be unsafe.
By \autoref{lem:elem_back_change}, 
\begin{align*}
    \sum_{w \in K} (\Delta \Phi_w + \Delta \Psi_w) 
        & \leq \sum_{w \in K'} (\Delta \Phi_w + \Delta \Psi_w) 
            \leq 3 \beta \cdot |K'| 
            = 3 \beta \cdot 2^{p(z) - \fk+1} \\
        & \leq 2^{p(z)} \leq \pi(z). && (\text{by \cref{fact:constants}})
\end{align*}

As the only elements that may change their budgets are $z$ and elements from $K$,
we have $\Delta \DLM + \Delta \Psi + \sum_{w \neq z} \Delta \Phi_w
= \Delta \DLM + \sum_{w \in K} \Delta \Psi_w + \Delta \Psi_z + \sum_{w \in K}
\Delta \Phi_w \leq 2 \cdot \pi(z) + \Delta \Psi_z \leq 2 \cdot \pi(z)$. The last inequality
follows as $\Psi_z$ drops to $0$ when $z$ is moved to the list front.
\end{proof}


Now we may split the cost of \DLM in a single step into parts incurred by Lines
\ref{line:dlma_first}--\ref{line:dlma_budget_increase} and
Lines~\ref{line:dlma_last_line_2}--\ref{line:dlma_last_line_3}, and bound them
separately.

\begin{lemma}
\label{lem:amortized_cost_56}
Whenever $\DLM$ executes Lines~\ref{line:dlma_last_line_2}--\ref{line:dlma_last_line_3}
of \Cref{alg:dlma}, $\Delta \DLM + \Delta \Phi +\Delta \Psi \leq 0$.
\end{lemma}

\begin{proof}
Let $z$ be the element moved in \autoref{line:dlma_last_line_3}. 
\autoref{line:dlma_last_line_2} guarantees that $b(z) \geq \pos(z)$
and \autoref{obs:budget_invariant} implies $b(z) \leq (3/2) \cdot \pos(z)$.
The value of $\Phi_z$ before the movement is then
\begin{align*}
    \Phi_z
    & \geq \min \{\fa \cdot b(z), \, \fb \cdot \pos(z) - \fc \cdot b(z) \}\\
    & \geq \min \left\{\fa \,,  \fb - (3/2) \cdot \fc \right\} \cdot \pos(z) \\
    & \geq 2 \cdot \pos(z). & \text{(by \cref{fact:constants})}
\end{align*}
When $z$ is moved to the list front, potential $\Phi_z$ drops to $0$, and thus
$\Delta \Phi_z \leq -2\cdot \pos(z)$. Hence, using \autoref{lem:fetch_cost},
$\Delta \DLM + \Delta \Phi +\Delta \Psi\leq 2 \cdot \pos(z) + \Delta \Phi_z \leq 0$.
\end{proof}

\begin{lemma}
\label{lem:amortized_cost}
Fix any step and consider its first part, where $\DLM$ 
pays for its access and movement costs, whereas $\OFF$ pays for 
its access cost. Then,
$\Delta \DLM + \Delta \Phi +\Delta \Psi \leq 
(3 + \fa) \cdot 2^{\fk+1} \cdot \Delta \OFF = 
O(r) \cdot \Delta \OFF$.
\end{lemma}

\begin{proof}
Let $R = \{x, y_1, \dots, y_{s-1} \}$ be the requested set, where $s \leq r$ and
$\pos(x) < \pos(y_i)$ for any $i \in [s-1]$. Let $\Phi_x$
denote the value of the potential just before the request.
It suffices to analyze the amortized cost of \DLM in
Lines~\ref{line:dlma_first}--\ref{line:dlma_budget_increase} as the cost in the
subsequent lines is at most~$0$ by \autoref{lem:amortized_cost_56}.
In these lines:
\begin{itemize}
\item \DLM pays $\pi(x)$ for the access.
\item \DLM performs the operation $\fetch(x)$, whose 
amortized cost is, by \autoref{lem:fetch_cost}, 
at most $2 \cdot \pos(x) - \Phi_x$.
\item The budget of $y_i$ grows by $\Delta b(y_i) = \pi(x) / s$ for each $i \in [s-1]$. 
As these elements do not move (within Lines~\ref{line:dlma_first}--\ref{line:dlma_budget_increase}), 
$\Delta \Psi_{y_i} = 0$. 
\end{itemize}
Thus, we obtain 
\begin{equation}
\label{eq:alg_cost_1}
    \Delta \DLM + \Delta \Phi +\Delta \Psi
    \leq 3 \cdot \pos(x) - \Phi_x + \sum_{i \in [s-1]} \Delta \Phi_{y_i}.
\end{equation}
As elements $y_i$ do not move (within
Lines~\ref{line:dlma_first}--\ref{line:dlma_budget_increase})),
the change in $\Phi_{y_i}$ can be induced only by the change in the budget 
of $y_i$.
Let $u \in R$ be the element with the smallest position on the list of \OFF, i.e.,
$\Delta \OFF = \pos^*(u)$.
We consider three cases. 

\begin{itemize}
\item $p(x) \leq p^*(u) + \fk$. 

Then $\pi(x) \leq 2^{p(x)+1} \leq 2^{\fk+1} \cdot 2^{p^*(u)} \leq 2^{\fk+1} \cdot \pos^*(u)
= 2^{\fk+1} \cdot \Delta \OFF$.
Note that 
$\sum_{i \in [s-1]} \Delta \Phi_{y_i} \leq 
\sum_{i \in [s-1]} \alpha \cdot \Delta b(y_i) =
(s-1) \cdot \alpha \cdot \pos(x) / s < \alpha \cdot \pos(x)$.
By \autoref{lem:non_negative}, $\Phi_x \geq 0$, and thus using \eqref{eq:alg_cost_1},
\begin{align*}
    \Delta \DLM + \Delta \Phi +\Delta \Psi 
    & < 3 \cdot \pos(x) + \alpha \cdot \pi(x)
    \leq (3 + \alpha) \cdot 2^{\fk+1} \cdot \Delta \OFF.
\end{align*}

\item $p(x) \geq p^*(u) + \fk + 1$ and $u = x$.

In this case, 
    $\Phi_x \geq \fb \cdot \pos(x) - \fc \cdot b(x) \geq (\fb - (3/2)
    \cdot \fc) \cdot \pos(x)$ (cf.~\autoref{obs:budget_invariant}). 
    By plugging this bound to \eqref{eq:alg_cost_1}, we obtain 
\begin{align*}
\Delta \DLM + \Delta \Phi +\Delta \Psi 
& \leq 3 \cdot \pos(x) + 
    (\fb - (3/2) \cdot \fc) \cdot \pos(x) + \alpha \cdot \pos(x) \leq 0,
\end{align*}
where the last inequality follows as $\beta \geq 3 + (3/2) \cdot \gamma + \alpha$
by \cref{fact:constants}.

\item $p(x) \geq p^*(u) + \fk + 1$ and $u = y_j$ for some $j \in [s-1]$. 

    Recall that $\pos(x) < \pos(y_i)$, and thus 
    $p(y_j) \geq p(x)$. Hence, $p(y_j) \geq p^*(y_j) + \fk + 1$.
    In such a case,
    \begin{align*}
        \sum_{i \in [s-1]} \Delta \Phi_{y_i} 
        & = \Delta \Phi_{y_j} + \sum_{i \in [s-1] \setminus \{j\} } \Delta \Phi_{y_i} \\
        & \leq -\gamma \cdot \Delta b(y_j) + 
        \sum_{i \in [s-1] \setminus\{j\}} \alpha \cdot \Delta b(y_i) \\
        & = - \gamma \cdot \pos(x) / s + (s-2) \cdot \alpha \cdot \pos(x) / s \\
        & < (\alpha - \gamma/r) \cdot \pos(x). && \text{(as $s \leq r$)} 
    \end{align*}
    Plugging the bound above and $\Phi_x \geq 0$ to 
    \eqref{eq:alg_cost_1} yields
    \[
    \Delta \DLM + \Delta \Phi+\Delta \Psi 
    \leq (3 + \fa - \fc/r) \cdot \pos(x) \leq 0,
    \]
    where the last inequality again follows by \cref{fact:constants}. 
    \qedhere
\end{itemize}
\end{proof}

\begin{theorem} \label{thm:DLMSO}
$\DLM$ is $O(r)$-competitive in the static scenario.
\end{theorem}

\begin{proof}
Fix any input $\I$ ad any offline solution \OFF that maintains a fixed permutation.
For any step $t$, let $\Phi^t$ and $\Psi^t$ denote the total 
potentials right after step $t$, while $\Phi^0$ and $\Psi^0$ 
be the initial potentials.
By \autoref{lem:amortized_cost},
\begin{align}
\label{eq:single_step}
    \DLM_t(\I) & + \Phi^t + \Psi^t - \Phi^{t-1}- \Psi^{t-1} = O(r) \cdot \OFF_t(\I),
\end{align}
where $\DLM_t(\I)$ and $\OFF_t(\I)$ denote the costs of \DLM and \OFF in
step $t$, respectively. By summing \eqref{eq:single_step} over all $m$ steps of
the input, we obtain $\DLM(\I) + \Phi^m +\Psi^m - \Phi^0 - \Psi^0 \leq O(r)
\cdot \OFF(\I)$. As $\Phi^m + \Psi^m \geq 0$,
\[
    \DLM(\I) \leq O(r) \cdot \OFF(\I) + \Phi^0 + \Psi^0 .
\]
Note that the initial potentials might be non-zero as in the static 
scenario \OFF starts in its permutation which might be different 
from $\pi_0$. That said, both initial potentials can be universally
upper-bounded by the amount independent of $\I$, and thus 
\DLM is $O(r)$-competitive.
\end{proof}


\section{Analysis in the dynamic scenario}
\label{sec:dynamic}

To analyze \DLM in the dynamic scenario, we first 
establish an~offline approximation of \OPT that could be handled using 
our potential functions.

We say that an algorithm is \emph{move-to-front based ($\MTF$-based)} if, in
response to request~$R$, it chooses exactly one of the elements from $R$, brings it
to the list front, and does not perform any further actions. We denote
the class of all such (offline) algorithms by \MTFB. The proof of the following
lemma can be found in the appendix.

\begin{lemma}
\label{lem:mtf_opt}
For any input $\I$, there exists an (offline) algorithm 
$\OFF^* \in \MTFB$, such that $\OFF^*(\I) \leq 4 \cdot \OPT(\I)$.
\end{lemma}

We now analyze the second stage of a step, where an offline algorithm \OFF from
the class \MTFB reorders its list. 

\begin{lemma}
\label{lem:delta_phi_when_opt_moves}
Assume $\OFF \in \MTFB$. 
Fix any step and consider its second stage, 
where $\OFF$ moves some element $z$ to the list front. 
Then, $\Delta \Phi +\Delta \Psi  = O(r^2) \cdot \Delta \OFF$.
\end{lemma}

\begin{proof}
We may assume that initially $\pos^*(z) \geq 2$, as otherwise
there is no change in the list of \OFF and the lemma follows trivially.

Apart from element $z$, the only elements that change their positions are
elements that originally preceded $z$: their positions are incremented. By
\autoref{lem:opt_increments_elems}, the potential change associated with these
elements is non-positive. 

Thus, $\Delta \Phi + \Delta \Psi \leq \Delta \Phi_z + \Delta \Psi_z$.
Element~$z$ is transported by $\OFF$ from position $\pos^*(z)$ to position $1$,
i.e., $\Delta \OFF = \pos^*(z) - 1 \geq \pos^*(z) / 2$ as 
we assumed $\pos^*(z) \geq 2$.
Thus, to show the lemma it suffices to show that 
$\Delta \Phi_z = O(r^2) \cdot \pos^*(z)$ 
and $\Delta \Psi_z = O(r^2) \cdot \pos^*(z)$. 
We bound them separately.


\begin{itemize}
\item Note that $p^*(z)$ may only decrease. 
If initially $p^*(z) \leq p(z) - \fk - 1$, then 
$\Phi_z = \beta \cdot \pos(z) - \gamma \cdot b(z)$ before 
and after the movement of $z$, and thus 
$\Delta \Phi_z = 0$. 
Otherwise, $p^*(z) \geq p(z) - \fk$,
which implies 
$\pos(z) < 2^{p(z)+1} \leq 2^{\fk+1} \cdot 2^{p^*(z)}
 \leq 2^{\fk+1} \cdot \pos^*(z)$. 
In such a case, 
\[
    \Delta \Phi_z 
    \leq \beta \cdot \pos(z) - \gamma \cdot b(z) - \alpha \cdot b(z)  
    \leq \beta \cdot \pi(z) 
    \leq \beta \cdot 2^{\fk+1} \cdot \pi^*(z) = O(r^2) \cdot \pi^*(z).
\]
\item 
Similarly, if initially $p^*(z) \leq p(z) - \fk$, then 
$\Psi_z = 2\beta \cdot q(z)$ before 
and after the movement of $z$, and thus 
$\Delta \Psi_z = 0$. 
Otherwise, $p^*(z) \geq p(z) - \fk + 1$ and
which implies 
$\pos(z) < 2^{p(z)+1} \leq 2^{\fk} \cdot 2^{p^*(z)}
 \leq 2^{\fk} \cdot \pos^*(z)$.
In such a case
\[
    \Delta \Psi_z 
    \leq 2 \beta \cdot q(z) - 0 
      \leq 2 \beta \cdot \pi(z) 
    \leq 2 \beta \cdot 2^{\fk} \cdot \pi^*(z)
    = O(r^2) \cdot \pi^*(z).
    \qedhere
\]
\end{itemize}
\end{proof}


\begin{theorem} \label{thm:DLMDO}
$\DLM$ is strictly $O(r^2)$-competitive in the dynamic scenario.
\end{theorem}

\begin{proof}
The argument here is the same as for 
\autoref{thm:DLMSO}, but this time we sum the guarantees provided 
for the first stage of a step (\autoref{lem:amortized_cost})
and for the second stage of a step (\autoref{lem:delta_phi_when_opt_moves}).
This shows that for any offline algorithm $\OFF \in \MTFB$ and any
input $\I$ it holds that
\begin{equation}
\label{eq:dlm_competitive_against_mtf}
    \DLM(\I) \leq O(r^2) \cdot \OFF(\I) + \Phi^0 + \Psi^0 .
\end{equation}
For the dynamic scenario, the 
initial permutations of \DLM and \OFF are equal, and hence
the initial potential $\Phi^0 + \Psi^0$ is zero.
As \eqref{eq:dlm_competitive_against_mtf} holds 
against arbitrary $\OFF \in \MTFB$, it holds also against $\OFF^*$ 
which is the $4$-approximation of \OPT (cf.~\autoref{lem:mtf_opt}). 
This implies that 
\[
    \DLM(\I) \leq O(r^2) \cdot \OFF^*(\I) \leq O(4 \cdot r^2) \cdot \OPT(\I) ,
\]
which concludes the proof.
\end{proof}


\section{Final remarks}

In this paper, we studied achievable competitive ratios for the online \MSSC
problem. We closed the gaps for deterministic polynomial-time static scenarios
and tighten the gaps for deterministic dynamic scenarios. Still, some intriguing
open questions remain, e.g., the best randomized algorithm for the dynamic
scenario has a competitive ratio of $O(r^2)$, while the lower bound is merely a
constant. 

Another open question concerns a generalization of the MSSC problem where each
set~$R_t$ comes with a~covering requirement $k_t$ and an algorithm is charged
for the positions of the first $k_t$ elements from $R_t$ on the list (see,
e.g.,~\cite{BaBaFT21}). The only online results so far 
are achieved in the easiest, learning scenario~\cite{FoLiPS20}.

\bibliography{references}

\appendix

\section{Tightness of the analysis}
\label{sec:lower}

In this section, we show that our analysis for \DLM is tight by presenting an
input sequence~$\I$ for which $\DLM(\I) \geq \Omega(r^2) \cdot \OPT(\I)$. We
show that such a bound holds even for the parameterized version of \DLM that we
describe below.

Recall that when \DLM processes a request $\{x, y_1, y_2, \dots, y_{s-1} \}$
(where $\pos(x) < \pos(y_i)$), it increases the budgets of all elements $y_i$ by
$\ell/s$ (cf.~\autoref{line:dlma_budget_increase} of \Cref{alg:dlma}). There,
$\ell$ is the initial value of $\pos(x)$ right before the request. If
we changed $s$ to $r$, the analysis of \DLM would be essentially the same, with
the disadvantage that then \DLM would need to know $r$ upfront. 

One may wonder whether a better ratio could be achieved by simply changing the
denominator of $r$ to another value. This motivates the following
parameterization of \DLM, called  $\DLM_c$. For a positive integer $c$, $\DLM_c$
is defined as \DLM, but we replace the assignment in
\autoref{line:dlma_budget_increase} of \Cref{alg:dlma} with
\[
    b(y_i) \leftarrow b(y_i) + \ell /c .
\]
For our lower bound, we make an additional assumption: we assume that if $k > 1$
elements have budgets exceeding their current positions 
(cf.~Lines \ref{line:dlma_last_line_2}--\ref{line:dlma_last_line_3} of \Cref{alg:dlma}),
then $\DLM_c$ fetches them to the first $k$ positions preserving their relative order.

\begin{theorem}
\label{thm:lower_bound}
For any positive integer $c$, the competitive ratio of $\DLM_c$ in the 
dynamic scenario is $\Omega(r^2)$.
\end{theorem}

In our construction, we set the number of elements~$n$ to be $n = r^2 + c \cdot
r$. The adversary employs the following play consisting of $c+1$ requests, where
each request contains the last $r$ elements in the current permutation of
$\DLM_c$. We call such play a \emph{phase}.

\begin{lemma}
\label{lem:lower_bound}
Assume that, at some time, the budgets of all elements are zero. 
Let $A$ be the set of the last $c+r$ elements 
on the list of $\DLM_c$. 
After $\DLM_c$ serves the phase:
\begin{itemize}
\item the budgets of all elements are again zero,
\item the first $c+r$ positions contain elements of $A$ shifted cyclically, and
\item all list elements not in $A$ increase their positions by $c+r$. 
\end{itemize}
Furthermore, the cost of $\DLM_c$ in such a phase is at least 
$(c+r) \cdot (n-r)$.
\end{lemma}

\begin{proof}
We name elements of set $A$ as $A = \{a_1, a_2, \dots, a_{c+1}, b_1, b_2, \dots
b_{r-1}\}$. Due to each of $c+1$ requests, \ALG fetches an element $a_i$ from
position $n-r+1$ to the list front and increases budgets of all elements $b_j$
by $(n-r+1)/c$. Within the first $c$ requests, these budgets are at most $c
\cdot (n-r+1) / c = n-r+1$, and are thus smaller than the positions of the
corresponding elements $b_j$. However, after the $(c+1)$-th request, the budget
of any element $b_j$ will become equal to 
\[
    (c+1) \cdot \frac{n - r + 1}{c} 
    \geq (c+1) \cdot \frac{n - r}{c} 
    = n + \frac{r^2 - r}{c} 
    \geq n.
\]
Therefore, all elements $b_1, b_2, \dots, b_{r-1}$ are fetched to the list
front, and their budgets are reset to zero. (The budgets of all remaining
elements are always zero within the phase.) Once the phase ends, the list front
contains elements $b_1, b_2, \dots, b_{r-1}, a_1, a_2, \dots, a_{c+1}$, in this
order.

It remains to lower-bound the cost of $\DLM_c$ in a phase. We neglect its access
cost; its reordering cost consists of fetching $c+r$ elements, each from
position $n-r+1$ or further, i.e., its total cost is at least $(c+r) \cdot
(n-r)$.
\end{proof}

\begin{proof}[Proof of \autoref{thm:lower_bound}]
Consider an input sequence $\I$ consisting of $m$ phases, where $m$ is a~large
number. As $n$ is divisible by $c+r$, we can partition all list elements into
$r$ disjoint sets $A_1, A_2, \ldots, A_r$, each of size~$c+r$, that will 
be requested in the cyclic order in the consecutive phases of $\I$.

We bound $\OPT(\I)$ by presenting an offline solution \OFF for $\I$ with a small
cost. Let $A'_i$ contain every $(r-1)$-th element from the original ordering of
$A_i$, i.e., $|A'_i| = \lfloor |A_i| / (r-1) \rfloor = \lfloor (c+r) / (r-1)
\rfloor$. Let $A^* = \bigcup_{i=1}^{r} A'_i$. At the very beginning, \OFF moves
all elements from $A^*$ to the first $|A^*| = r \cdot \lfloor (c+r) / (r-1)
\rfloor \leq 2 c + 3 r$ positions. The cost of this operation becomes negligible
in comparison to the remaining cost for a sufficiently large number $m$ of phases.
Throughout the whole input $\I$, \OFF will keep elements of $A^*$ on the first $|A^*|$
positions, but it will change their relative order as discussed below.

Consider now any phase, in which elements from $A_i$ are requested. At its
beginning, by~\autoref{lem:lower_bound}, elements of $A_i$ are in the cyclically
shifted original order. Hence, $A^*$ contains one element $x^*$ among those that
$A_i$ has now at the last $r-1$ positions of the list of $\DLM_c$. 
These $r-1$ elements are present in all requested sets of a phase. To serve
them, before the phase starts, \OFF fetches $x^*$ to the first position. Then, its
cost within a phase is at most $|A^*| - 1 \leq 2 c + 3 r - 1$ for reordering and
$c+1$ for accesses, i.e., $3 c + 3 r$ in total.

Recall that by \autoref{lem:lower_bound}, the cost of $\DLM_c$ in a phase 
is at least $(c+r) \cdot (n-r)$. 
Hence, for each phase, the $\DLM_c$-to-\OFF ratio is 
\begin{align*}
    \frac{\DLM_c}{\OPT} 
    & \geq \frac{(c+r) \cdot (n-r)}{3 c + 3r} =  \frac{n-r}{3} \geq \frac{r^2}{3}.
\end{align*}
Summing over all phases yields the desired lower bound.
\end{proof}


\section{Omitted proofs}

\begin{proof}[Proof of \autoref{lem:mtf_opt}]
Based on the actions of $\OPT$ on $\I = (\pi_0, R_1, \ldots, R_m)$, we may
create an input $\J = (\pi_0, R_1', \ldots, R_m')$ where $R_i'$ is a singleton
set containing exactly the element from~$R_i$ that $\OPT$ has nearest to the
list front.

Clearly, $\OPT(\J) \leq \OPT(\I)$. Note that $\J$ is an instance of the list
update problem. Thus, if we take an algorithm $\MTF$ for the list update problem
(which brings the requested element to the list front), then it holds that
$\MTF(\J) \leq 4 \cdot \OPT(\J)$~\cite{SleTar85}. (The result of~\cite{SleTar85}
shows the competitive ratio of $2$, but the list update model ignores reordering
costs. However, for \MTF the reordering costs are equal to access costs
minus~$1$ and hence, taking them into account at most doubles the competitive
ratio.)

Let now $\OFF^*$ be an offline algorithm that, on input $\I$, performs the same
list reordering as $\MTF(\J)$. Clearly, $\OFF^* \in \MTFB$. While the
reordering cost of $\OFF^*$ on $\I$ coincide with that of $\MTF$ on $\J$,
its access cost can be only smaller. Thus, $\OFF^*(\I) \leq \MTF(\J)$.

Summing up, we obtain $\OFF^*(\I) \leq \MTF(\J) \leq  4 \cdot \OPT(\J) \leq 4
\cdot \OPT(\I)$, which concludes the proof. 
\end{proof}

\end{document}